\documentclass[conference,a4paper]{IEEEtran}
\usepackage{cite}
\usepackage{amsthm}
\usepackage[cmintegrals]{newtxmath}
\usepackage{subcaption}
\usepackage{graphicx}
\usepackage{epstopdf}
\usepackage[detect-weight=true, binary-units=true]{siunitx}
\usepackage[inline]{enumitem}
\usepackage[boxruled,norelsize]{algorithm2e}
\usepackage{tabularx}
\usepackage{multirow}
\usepackage{booktabs}
\usepackage{diagbox}

\makeatletter
\newcommand{\removelatexerror}{\let\@latex@error\@gobble}
\makeatother

\newtheorem{theorem}{Theorem}

\IEEEoverridecommandlockouts
\IEEEpubid{\makebox[\columnwidth]{\copyright~Copyright 2018
		IEEE \hfill} \hspace{\columnsep}\makebox[\columnwidth]{ }}

\begin{document}
\title{A Fast Blind Impulse Detector for Bernoulli-Gaussian Noise in Underspread Channel}

\author{\IEEEauthorblockN{Bin~Han and Hans~D.~Schotten}
	\IEEEauthorblockA{Technische Universit\"at Kaiserslautern\\
		Institute of Wireless Communication\\
		Paul-Ehrlich-Stra\ss e 11, 67663 Kaiserslautern, Germany\\
		Email: \{binhan, schotten\}@eit.uni-kl.de}
}

\maketitle

\begin{abstract}
Impulsive noises widely existing in various channels can significantly degrade the performance and reliability of communication systems. The Bernoulli-Gaussian (BG) model is practical to characterize noises in this category. To estimate the BG model parameters from noise measurements, a precise impulse detection is essential. In this paper, we propose a novel blind impulse detector, which is proven to be fast and accurate for BG noise in underspread communication channels.
\end{abstract}

\IEEEpeerreviewmaketitle

\section{Introduction}
%The Bernoulli-Gaussian (BG) model has been widely applied on impulsive noises in various communication systems, including ultra wide-band (UWB) systems \cite{sharma2016sparsity}, Wireless Local Area Networks (WLAN) \cite{sanchez2007interference}, digital subscriber line (DSL) networks \cite{henkel1994wideband} and power line communication (PLC) systems \cite{han2017noise}. Compared to other common impulsive noise models such like the Middleton's Class-A (MCA) model \cite{middleton1979procedures} and the symmetric $\alpha$-stable (S$\alpha$S) model \cite{nikias1995signal}, the BG model stands out with its compatibility to different noise bandwidths, while simultaneously exhibiting a heavy-tailed probability density function (PDF) with a simple closed-form expression. Moreover, it can be easily extended to the Markov-Gaussian model to characterize noise bursts \cite{shongwe2015study}.

Impulsive noises are widely observed in various communication systems, including ultra wide-band (UWB) systems \cite{sharma2016sparsity}, wireless local area networks (WLAN) \cite{sanchez2007interference}, digital subscriber line (DSL) networks \cite{henkel1994wideband} and power line communication (PLC) systems \cite{han2017noise}. Due to their non-stationary nature and high peak power, they can significantly degrade the performance and reliability of communication systems. Such impacts can become critical in urban and industrial environments where \begin{enumerate*}
	\item frequent mechanical switching operations and vibrations are present to produce dense impulsive noises, and
	\item ultra-high reliability and ultra-low latency are expected in short-range wired/wireless communication applications.
\end{enumerate*} 
Various use cases of this kind have been addressed in scopes of both the Fifth Generation (5G) mobile networks \cite{osseiran2014scenarios} and new advanced industrial communication solutions \cite{bockelmann2017hiflecs}. In these contexts, techniques of modeling impulsive noises are required as a major tool to define channel models and thereby develop robust communication systems.

The Bernoulli-Gaussian (BG) model has been widely applied on impulsive noises. Compared to other common impulsive noise models such like the Middleton's Class-A (MCA) model \cite{middleton1979procedures} and the symmetric $\alpha$-stable (S$\alpha$S) model \cite{nikias1995signal}, the BG model stands out with its compatibility to different noise bandwidths, while simultaneously exhibiting a heavy-tailed probability density function (PDF) with a simple closed-form expression. Moreover, it can be easily extended to the Markov-Gaussian model to characterize noise bursts \cite{shongwe2015study}.

However, despite these superiorities, the deployment of BG model in communications and signal processing has been limited by the insufficient study on its parameter estimation, or more specifically, on the impulse detection. Unlike the MCA/S$\alpha$S models that describe the overall statistics of the mixed noise, the BG model separates impulses from the background noise, so that its parameter estimation relies on accuracy impulse detection. Unfortunately, most existing methods of blind impulse detection for BG processes, if not all, either suffer from high order of computational complexity, or highly rely on the initial guess to avoid convergence at local extremes that may bias far away from the ground truth. 

In this paper, focusing on the particular but common case of underspread channels, we propose a novel blind BG impulse detector, which is fast, accurate and reliable in wide ranges of impulse rate and impulse-to-background power ratio. The rest part of this paper is organized as follows: First,  in Sec.~\ref{sec:model} we setup the BG model for impulsive noises, discuss about its approximation in underspread channels, and analyze the detection model. Then in Sec.~\ref{sec:review} we briefly review the existing methods of BG impulse detection. Subsequently, in Sec. \ref{sec:methods} we introduce our proposed approach, including a novel iterative algorithm, a robust Gaussian estimator and a sparsity-sensitive initializing method. To the end, some numerical simulation results are presented in Sec.\ref{sec:simulation}, before we conclude our work and provide some outlooks in Sec.~\ref{sec:conclusion}.

\section{System Model}\label{sec:model}
\IEEEpubidadjcol 
\subsection{Impulsive Noises as Bernoulli-Gaussian Processes}
A BG process $X(\rho,\sigma_1^2,\sigma_2^2)$ switches randomly between two independent Gaussian states, the switching behavior is determined by an independent Bernoulli process $\Phi$:
\begin{equation}
	\begin{cases}
		\phi_n\sim B(1,\rho),\\
		(x_n|\phi_n)\sim\mathcal{N}\left(0,\sigma_1^2+\phi_n\sigma_2^2\right).
		\label{equ:bg_model}
	\end{cases}
\end{equation}
In the context of impulsive noise, $\rho\in[0,1]$ is the \textit{impulse rate}, $\sigma_1^2$ is the \textit{background noise power}, $\sigma_2^2>\sigma_1^2$ is the \textit{impulsive noise power}, and $n\in\mathbb{N}$ is the index of samples. Note that every observation $x_n$ is {independently and identically distributed (i.i.d.)} with respect to the PDF
\begin{align}
	\begin{split}
		f_X(x)&=p_\Phi(0)f_X(x|\Phi=0)+p_\Phi(1)f_X(x|\Phi=1)\\
		&=\frac{\rho}{\sqrt{2\pi\sigma_1^2}}e^{-\frac{x^2}{2\sigma_1^2}}+\frac{1-\rho}{\sqrt{2\pi\left(\sigma_1^2+\sigma_2^2\right)}}e^{-\frac{x^2}{2\left(\sigma_1^2+\sigma_2^2\right)}}.
	\end{split}
	\label{equ:bg_pdf}
\end{align}

In most literatures such as~\cite{kormylo1982maximum,lavielle1993bayesian,champagnat1996unsupervised,soussen2011bernoulli,mendel2013optimal}, the observation distortion is considered, so that a linear time-invariant (LTI) system $H$ is introduced to filter the BG sequence $\mathbf\{x_0,x_1,\dots x_{N-1}\}$, and an extra error $\epsilon$ is added at the output, to generate the final observation sequence:
\begin{equation}
	y_n=\sum_{m=0}^{L}x_nh_{n-m}+\epsilon_n,
\end{equation}
where $h$ is the impulse response of $H$ with $L$ as its length, and $\epsilon\sim\mathcal{N}(0,\sigma_\epsilon^2)$ is Gaussian distributed. Usually, $h$ and $\sigma_\epsilon^2$ are considered as known so that $p(y_n|\phi_n)$ can be calculated with $p(x_n|\phi_n)$, $h$ and $\sigma_\epsilon^2$ straight-forward.

\subsection{Approximations in Underspread Channels}
Many channels with impulsive noises are reported to be generally underspread, such as UWB channels~\cite{raghavan2007capacity} and PLC channels~\cite{canete2008time}. In this case, $L$ is negligible with respect to the symbol interval of communication systems, so we can approximately consider that $h_n=\delta(n)$, and therefore $y_n=x_n+\epsilon_n$. As $\epsilon$ is a Gaussian noise independent from $X$, we have
\begin{equation}
	(y_n|\phi_n)\sim\mathbb{N}(0,\sigma_1^2+\sigma_\epsilon^2+\phi_n\sigma_2^2).
\end{equation}
This differs from Eq.~(\ref{equ:bg_model}) only with a known constant offset on the background noise power. Hence, it is convenient not to distinguish the observation error from background noise, but to simply assume that an uncontaminated observation sequence of $X$ is available, as usually done in the field of noise characterization e.g. \cite{zimmermann2000analysis,shongwe2015study,han2017noise}.

\subsection{Bayesian Estimation and Impulse Detection}
Taking the underspread assumption, the problem of BG parameter estimation can be represented as: given a finite sequence of observation $\mathbf{x}=\{x_0,x_1,\dots,x_{N-1}\}$, to estimate the most probable parameter vector $\theta=\left(\rho,\sigma_1^2,\sigma_2^2\right)$:
\begin{equation}
		\hat{\theta}=\arg\max_{\theta\in\Theta}f\left(\theta|\mathbf{x}\right).
\end{equation}
where $\Theta=[0,1]\times\mathbb{R}^+\times\mathbb{R}^+$ is the space of $\theta$. As the a posteriori PDF $f\left(\theta|x_0,x_1,\dots,x_{N-1}\right)$ is impractical to obtain, we naturally rely on the Bayesian method:
\begin{equation}
		f(\theta|\mathbf{x})=\frac{f(\mathbf{x}|\theta)f(\theta)}{f(\mathbf{x})}.
\end{equation}
Recalling that every observation $x_n$ is \textit{i.i.d.}, we have
\begin{align}
	\begin{split}
		\hat{\theta}&=\arg\max_{\theta\in\Theta}\frac{f(\mathbf{x}|\theta)f(\theta)}{f(\mathbf{x})}\\&	=\arg\max\limits_{\theta\in\Theta}f(\mathbf{x}|\theta)f(\theta)\\
		&=\arg\max\limits_{\theta\in\Theta}f(\theta)\prod_{n=0}^{N-1}f_X(x_n|\theta).
	\end{split}
	\label{sec:bayesian_theta}
\end{align}
where $f_X(x_n|\theta)$ can be calculated as Eq.~(\ref{equ:bg_pdf}). However, the a priori PDF $f(\theta)$ is still hard to directly obtain from $\mathbf{x}$ or $\hat{\theta}$.

Noticing that both Bernoulli and Gaussian processes are stationary (although BG processes are non-stationary), $\theta$ can be consistently estimated from the empirical statistics, if only the ground truth of Bernoulli sequence $\underline{\phi}=\{\phi_0,\phi_1,\dots,\phi_{N-1}\}$ is known:
\begin{align}
	\begin{cases}
		\lim\limits_{N\to\infty}\frac{1}{N}\sum\limits_{n=0}^{N-1}\phi_n=\rho;\\
		\lim\limits_{N\to\infty}\textrm{Var}\{x_n|\phi_n=0\}=\sigma_1^2;\\
		\lim\limits_{N\to\infty}\textrm{Var}\{x_n|\phi_n=1\}=\sigma_1^2+\sigma_2^2.
	\end{cases}
	\label{equ:ergodic}
\end{align}
More importantly, unlike $f(\theta)$, the a priori probability mass function of $\underline{\phi}$ can be simply written as
\begin{equation}
	p\left(\underline{\phi}\right)=\prod\limits_{n=0}^{N-1}\rho^{\phi_n}(1-\rho)^{1-\phi_n}.
	\label{equ:pdf_phi_sequence}
\end{equation}
This encourages to estimate $\underline{\phi}$ instead of directly estimating $\theta$. Similar to Eq.~(\ref{sec:bayesian_theta}) we have
\begin{align}
	\begin{split}
		\hat{\underline{\phi}}&=\arg\max_{\underline{\phi}\in\Omega}p\left(\underline{\phi}|\mathbf{x}\right)\\
		&=\arg\max_{\underline{\phi}\in\Omega}p\left(\underline{\phi}\right)\prod_{n=0}^{N-1}f_X\left(x_n|\underline{\phi}\right)\\
		&=\arg\max_{\underline{\phi}\in\Omega}p\left(\underline{\phi}\right)\prod_{n=0}^{N-1}f_X\left(x_n|\phi_n\right)
	\end{split}
	\label{equ:bayesian_phi_sequence}
\end{align}
where $\Omega=\{0,1\}^N$. When $N$ is large enough, we can estimate $\rho$, $\sigma_1^2$ and $\sigma_2^2$ with Eq.~(\ref{equ:bg_pdf}), then calculate $p\left(\underline{\phi}\right)$ with Eq.~(\ref{equ:pdf_phi_sequence}), and $f_X\left(x_n|\underline{\phi}\right)$  as
\begin{equation}
	f_X\left(x_n|\phi_n\right)=\frac{1}{\sqrt{2\pi(\sigma_1^2+\phi_n\sigma_2^2)}}e^{-\frac{x_n^2}{2\left(\sigma_1^2+\phi_n\sigma_2^2\right)}}
\end{equation}

Thus, the problem of estimating $\theta$ from a continuous space $\Theta$ is converted to an impulse-detection problem, where an optimum should be selected from $2^N$ different binary sequences.

\section{Existing BG Impulse Detection Methods}\label{sec:review}
When the impulse rate $\rho$ is low and the impulse-to-background power ratio ${\sigma_2^2}/{\sigma_1^2}$ is high, the detection of high-powered impulses from BG noise can be easily accomplished through rejecting outliers with robust statistics and simple thresholding, e.g. the approach reported in \cite{han2013online}. However, impulsive noises do not always fulfill both the premises simultaneously. Aiming at a universal solution for a general $\Theta$, the Bayesian approach in Eq.~(\ref{equ:bayesian_phi_sequence}) is preferred. 

Due to the absence of gradient information, Eq.~(\ref{equ:bayesian_phi_sequence}) cannot be analytically solved. On the other hand, a full-search for the optimal $\hat{\underline{\theta}}$ in the space $\{0,1\}^N$ is impractical due to its huge time complexity of exponential order. Therefore, heuristic optimization algorithms appear attractive for this problem. The most classical method of this kind for BG processes is provided by \textit{Kormylo} and \textit{Mendel}, who proposed their famous \textit{single most likely replacement (SMLR)} algorithm in the early-1980s \cite{kormylo1982maximum}, which begins with an arbitrary initial guess of $\hat{\underline{\phi}}$ and iteratively update it. In each loop, it updates one and only one sample in $\hat{\underline{\phi}}$, which maximizes the updated likelihood function. The algorithm keeps iterating in loops to its convergence, i.e. until no single-sample update of $\hat{\underline{\phi}}$ can further raise the resulted likelihood function. The SMLR is proven to be highly practical due to its simple and efficient iterative implementation, which decreases the time complexity from exponential $\mathcal{O}(2^N)$ to polynomial $\mathcal{O}(N^3)$. 

However, the weakness of SMLR is also conspicuous, that it can easily end up with a local convergence instead of global optimum. Therefore, its performance relies so highly on the initial guess, that it can hardly be deployed as a blind detector alone, but only under supervision of another initializing estimator. Besides, the SMLR was designed for BG sequence deconvolution, where $\theta$ is known so that $f(\underline{\phi})$ can be directly calculated without applying Eq.~(\ref{equ:ergodic}). Under this condition, noise measurements can be broken into small subsequences, each with $N$ up to several thousands, for which the SMLR's time complexity of $\mathcal{O}(N^3)$ is reasonably satisfying. In contrast, for our goal of blind BG model estimation under discussion here, $\theta$ is unknown but must be online estimated and updated with respect to $\hat{\underline{\phi}}$ in every iteration. To guarantee the validity of ergodic estimation in Eq.~(\ref{equ:ergodic}), especially in the cases with very small values of $\rho$ (for instance, PLC noises with weak disturbance are reported to have impulse rate down to $1.35\times10^{-5}$ \cite{zimmermann2000analysis}), a huge observation length $N$ becomes essential, and the SMLR with cubic time complexity appears computationally expensive. Therefore, the SMLR is inappropriate for blind BG model estimation of highly sparse impulsive noises.

In the mid-1990s, \textit{Champagnat} et al. enhanced SMLR with further improved numerical efficiency and memory requirement \cite{champagnat1996unsupervised}. Recently in the context of sparse signal restoration, \textit{Soussen} et al. have also adopted SMLR as the so-called single best replacement (SBR), which is reported to be fast and stable\cite{soussen2011bernoulli}. Nevertheless, both the variants remain with the same order of time complexity, and do not overcome the problem of local convergence.

Apart from the SMLR algorithm, \textit{Lavielle} has shown that some classical Bayesian methods such as the maximum a posteriori (MAP), the marginal probability mode (MPM) and the iterative conditional mode (ICM) also give good performances in BG deconvolution \cite{lavielle1993bayesian}. Especially, the ICM is proven to be much faster than the others with its quadric time complexity, which enables its deployment on huge datasets. However, similar to the SMLR, the ICM also usually converges to local minimums, leading to an accuracy severely depending on the initial guess of impulse locations.

\section{Methods}\label{sec:methods}
Aiming at a fast impulse detection, we designed a novel approach, which we introduce in this section as follows.

\subsection{Iterative Threshold Shifting}
We start with a simpler problem where $\sigma_1^2$ and $\sigma_2^2$ are known. In this situation, the maximum likelihood impulse detector always has a thresholding behavior:
\begin{theorem}{}
	Given an observation segment $\mathbf{x}$ of underspread BG noise, where every sample $x_n$ is i.i.d. according to Eq.~(\ref{equ:bg_pdf}) and the Gaussian parameters $\left(\sigma_1^2, \sigma_2^2\right)$ are known, the output of the MLE defined by Eq.~(\ref{equ:bayesian_phi_sequence}) always fulfill
	\begin{equation}
		|x_m|\ge|x_n|\iff\hat{\phi}_m\ge\hat{\phi}_n
		\label{equ:threshold_theorem}
	\end{equation}
\end{theorem}

\begin{proof}
	Consider two different estimations of $\underline{\phi}$, namely $\underline{\alpha}$ and $\underline{\beta}$, respectively, which differ from each other at only two samples:
	\begin{equation}
		\begin{cases}
			\alpha_m=0,\quad\alpha_n=1;\\
			\beta_m=1,\quad\beta_n=0;\\
			\alpha_i=\beta_i\quad\forall i\notin\{m,n\}.
		\end{cases}
	\end{equation}
	Comparing their a posteriori probability densities we have:
	\begin{align}
	\begin{split}
		&\frac{p\left(\underline{\alpha}|\mathbf{x}\right)}{p\left(\underline{\beta}|\mathbf{x}\right)}\\&=\frac{f\left(\mathbf{x}|\underline{\alpha}\right)p\left(\underline{\alpha}\right)}{f\left(\mathbf{x}|\underline{\beta}\right)p\left(\underline{\beta}\right)}\\
	    =&\frac{f_X\left(x_m|\hat{\phi}_m=0\right)f_X\left(x_n|\hat{\phi}_n=1\right)}{f_X\left(x_m|\hat{\phi}_m=1\right)f_X\left(x_n|\hat{\phi}_n=0\right)}\\
	    =&\frac{e^{-\frac{x_m^2}{2\sigma_1^2}}e^{-\frac{x_n^2}{2\left(\sigma_1^2+\sigma_2^2\right)}}}{e^{-\frac{x_n^2}{2\sigma_1^2}}e^{-\frac{x_m^2}{2\left(\sigma_1^2+\sigma_2^2\right)}}}
	    =e^{\frac{1}{2}\left(x_n^2-x_m^2\right)\left(\frac{1}{\sigma_1^2}-\frac{1}{\sigma1^2+\sigma_2^2}\right)}
	\end{split}
	\label{equ:likelihood_ratio}
	\end{align}	
	As $p\left(\underline{\alpha}|\mathbf{x}\right)\ge 0$ and $p\left(\underline{\beta}|\mathbf{x}\right)\ge 0$, we know that
	\begin{equation}
		\begin{cases}
			p\left(\underline{\alpha}|\mathbf{x}\right) > p\left(\underline{\beta}|\mathbf{x}\right) &x_m^2<x_n^2;\\
			p\left(\underline{\alpha}|\mathbf{x}\right) = p\left(\underline{\beta}|\mathbf{x}\right) \quad \textrm{if} &x_m^2=x_n^2;\\
			p\left(\underline{\alpha}|\mathbf{x}\right) < p\left(\underline{\beta}|\mathbf{x}\right) & x_m^2>x_n^2.
		\end{cases}
		\label{equ:replacements_comparision}
	\end{equation}
	Given an arbitrary estimation $\hat{\underline{\phi}}$, an increase in the a posteriori probability density $p\left(\mathbf{x}|\hat{\underline{\phi}}\right)$ can be achieved by switching the values of an arbitrary assignment pair $\left(\hat{\phi}_m,\hat{\phi}_n\right)=(0,1)$ that $|x_m|>|x_n|$ to $(1,0)$. Keeping iteratively doing this until no such pair is available, Eq.~(\ref{equ:threshold_theorem}) is ensured to be valid.
\end{proof}

Thus, the MLE of $\underline{\phi}$ is converted into an optimal selection of a threshold $T$ that maximizes $p\left(\hat{\underline{\phi}}|\mathbf{x}\right)$ with
\begin{equation}
	\hat{\phi}_n=\begin{cases}
		1&|x_n|\ge T;\\
		0&\textrm{otherwise}.
	\end{cases}
\end{equation}

Calling back the mechanism of SMLR, where in each loop all $N$ different candidate updates of $\hat{\underline{\phi}}$ must be evaluated: now we know that only the candidates in a ''thresholded'' style of Eq.~(\ref{equ:threshold_theorem}) are meaningful. Based on this, we designed our \textit{iterative threshold shifting} (ITS) method, which is equivalent to the SMLR algorithm under the discussed conditions, as described in Fig.~\ref{fig:its}. Differing from the original LSMR, which has to compute the cost function for $N$ single replacement candidates in every loop, our ITS algorithm considers only two candidates in every loop, so that the time complexity is reduced from $\mathcal{O}(N^3)$ to $\mathcal{O}(N^2)$.

\begin{figure}[!h]
	\removelatexerror
	%	\begin{algorithmic}[1]
	\begin{algorithm}[H]
%		\footnotesize
		Given $\left(\sigma_1^2,\sigma_2^2\right)$, start with an initial threshold $T_0$\;
		\For(\qquad\emph{Main loop}){$i = 0$ to $N-1$}
		{
			$\hat{\phi}_n\gets\begin{cases}
				1 &|x_n|\ge T_i\\
				0 &\textrm{otherwise}
			\end{cases}$\;
			Estimate $\rho$ with $\hat{\underline{\phi}}$ and $\mathbf{x}$\;
			$P\gets f\left(\mathbf{x}|\hat{\underline{\phi}}\right)p\left(\hat{\underline{\phi}}\right)$\;
			Select $l:\quad x_l^2\le x_n^2, \quad \forall \hat{\phi_n}=1$\;
			Generate $\underline{\zeta}: \zeta_n=\begin{cases}
				0&n=l\\
				\hat{\phi}_n&\textrm{otherwise}
			\end{cases}$\;
%			Estimate $\rho_l,\sigma_{1,l}^2,\sigma_{2,l}^2$ with $\zeta$ and $\mathbf{x}$\;
			Estimate $\rho_l$ with $\underline{\zeta}$ and $\mathbf{x}$\;
			$P_l\gets f\left(\mathbf{x}|\underline{\zeta}\right)p\left(\underline{\zeta}\right)$\;
			Select $u:\quad x_u^2\ge x_n^2, \quad \forall \hat{\phi_n}=0$\;
			Generate $\underline{\eta}: \eta_n=\begin{cases}
			1&n=u\\
			\hat{\phi}_n&\textrm{otherwise}
			\end{cases}$\;
%			Estimate $\rho_u,\sigma_{1,u}^2,\sigma_{2,u}^2$ with $\eta$ and $\mathbf{x}$\;
			Estimate $\rho_u$ with $\underline{\eta}$ and $\mathbf{x}$\;
			$P_u \gets f\left(\mathbf{x}|\underline{\eta}\right)p\left(\underline{\eta}\right)$\;
	       \Switch(\qquad\emph{Threshold shifting}){$\max\{P,P_l,P_u\}$}{
				\Case{$P_l$}{
					%$\hat{\underline{\phi}}\gets\underline{\zeta}$
					$T_{i+1}\gets\min(x_n|\zeta_n=1)$
				}
				\Case{$P_u$}{
					%$\hat{\underline{\phi}}\gets\underline{\eta}$
					$T_{i+1}\gets\min(x_n|\eta_n=1)$
				}
				\Other{\Return$\hat{\underline{\phi}}$}
			}
		}
	\end{algorithm}
	\caption{The ITS algorithm with known Gaussian parameters}
	\label{fig:its}
\end{figure}

\subsection{Robust Gaussian Estimation}
When the Gaussian parameters $\sigma_1^2$ and $\sigma_2^2$ are unknown, they have to be online estimated. So the ITS algorithm must be modified as Fig.~\ref{fig:its_with_gaussian_estimation} shows. Nevertheless, as indicated by Eq.~(\ref{equ:ergodic}), the estimations are dependent on $\underline{\phi}$, so that Eq.~(\ref{equ:likelihood_ratio}) is rewritten as:
\begin{align}
	\begin{split}
		\frac{p\left(\underline{\alpha}|\mathbf{x}\right)}{p\left(\underline{\beta}|\mathbf{x}\right)}
		&=\frac{e^{-\frac{x_m^2}{2\sigma_{1,\alpha}^2}}e^{-\frac{x_n^2}{2\left(\sigma_{1,\alpha}^2+\sigma_{2,\alpha}^2\right)}}}{e^{-\frac{x_n^2}{2\sigma_{1,\beta}^2}}e^{-\frac{x_m^2}{2\left(\sigma_{1,\beta}^2+\sigma_{2,\beta}^2\right)}}}\times\prod_{i\in \mathcal{I}}\frac{\sigma_{1,\beta}}{\sigma_{1,\alpha}}e^{\frac{x_i^2}{2\sigma_{1,\beta}^2}-\frac{x_i^2}{2\sigma_{1,\alpha}^2}}\\
		&\times\prod_{j\in \mathcal{J}}\sqrt{\frac{\sigma_{1,\beta}^2+\sigma_{2,\beta}^2}{\sigma_{1,\alpha}^2+\sigma_{2,\alpha}^2}}e^{\frac{x_j^2}{2\left(\sigma_{1,\beta}^2+\sigma_{2,\beta}^2\right)}-\frac{x_j^2}{2\left(\sigma_{1,\alpha}^2+\sigma_{2,\alpha}^2\right)}}
	\end{split}
\end{align}
where $\left(\sigma_{1,\alpha}^2,\sigma_{2,\alpha}^2\right)$ and $\left(\sigma_{2,\alpha}^2,\sigma_{2,\beta}^2\right)$ are the Gaussian parameters estimated with respect to $\underline{\alpha}$ and $\underline{\beta}$, respectively; $\mathcal{I}$ and $\mathcal{J}$ are the index sets that
\begin{equation}
	\begin{cases}
		\alpha_i=\beta_i=0&\forall i\in\mathcal{I};\\
		\alpha_i=\beta_j=1&\forall j\in\mathcal{J}.
	\end{cases}
\end{equation}
\begin{figure}[!h]
	\removelatexerror
	%	\begin{algorithmic}[1]
	\begin{algorithm}[H]
%		\footnotesize
		Start with an initial threshold $T_0$\;
		\For(\qquad\emph{Main loop}){$i = 0$ to $N-1$}
		{
			Generate $\hat{\underline{\phi}}$\;
			Estimate $\rho,\sigma_1^2,\sigma_2^2$ with $\hat{\underline{\phi}}$ and $\mathbf{x}$\;
			$P\gets f\left(\mathbf{x}|\hat{\underline{\phi}}\right)p\left(\hat{\underline{\phi}}\right)$\;
			Select $l$,	generate $\underline{\zeta}$\;
			Estimate $\rho_l,\sigma_{1,l}^2,\sigma_{2,l}^2$ with $\underline{\zeta}$ and $\mathbf{x}$\;
			$P_l\gets f\left(\mathbf{x}|\underline{\zeta}\right)p\left(\underline{\zeta}\right)$\;
			Select $u$, generate $\underline{\eta}$\;
			Estimate $\rho_u,\sigma_{1,u}^2,\sigma_{2,u}^2$ with $\underline{\eta}$ and $\mathbf{x}$\;
			%Estimate $\rho_u$ with $\underline{\eta}$ and $\mathbf{x}$\;
			$P_u \gets f\left(\mathbf{x}|\underline{\eta}\right)p\left(\underline{\eta}\right)$\;
			\Switch(\emph{Threshold shifting}){$\max\{P,P_l,P_u\}$}{
				$$\hspace{-7cm}\dots$$
			}
		}
	\end{algorithm}
	\caption{The ITS algorithm with unknown Gaussian parameters}
	\label{fig:its_with_gaussian_estimation}
\end{figure}
This removes the superiority of thresholding-based impulse detection provided by Eq.~\ref{equ:replacements_comparision} and the ITS algorithm is therefore no more guaranteed to equal the SMLR. As a consequence, the converging speed will decrease while the error rate may increase. To overcome this problem, we need a \textit{robust Gaussian estimation} (RGE) technique that is only slightly impacted by a single pair of impulse assignments $\left(\hat{\phi}_m,\hat{\phi}_n\right)$, so that in every loop we can approximately consider the estimations of Gaussian parameters remain independent from the update candidate. In this work, we applied the well-known and widely used median absolute deviation (MAD):
\begin{equation}
	\textrm{MAD}(Z)=\underset{n=0,1,\dots,N-1}{\textrm{median}}\left(|z_i-\textrm{median}(\mathbf{z})|\right).
\end{equation}
According to \textit{Rosseeuw} and \textit{Croux} \cite{rousseeuw1993alternatives}, for Gaussian processes there is
\begin{equation}
	\sigma_Z\approx1.4826\times\textrm{MAD}(Z).
\end{equation}

\subsection{Sparsity-Sensitive Initialization}
An appropriate blind selection of the initial threshold $T_0$ is critical to the performance of the ITS algorithm, as when the starting point approaches towards the global optimum:
\begin{itemize}
	\item the risk of converging at a local maximum decreases;
	\item the time cost of computation is reduced.
\end{itemize}
It has been reported in \cite{han2013online} that the ''three-sigma rule''~\cite{grafarend2006linear} can be applied for effective blind detection of impulsive outliers when combined with MAD:
\begin{equation}
	T_0 = 3\tilde{\sigma}_X=4.4478\times\textrm{MAD}(X).
\end{equation}
This method has been derived as robust to sparse impulses. However, for general BG noises, when $\rho$ increases to a relatively high level (e.g. over 1 percent), the amplitudes of impulse samples start to exhibit a significant impact on its performance. Therefore, a correction with respect to the sparsity level of impulses is called for. Generally, $T_0$ should be raised to a higher level, when the impulses become sparser, i.e. when $\rho$ decreases and/or $\sigma_2^2/\sigma_1^2$ increases.

In \cite{hurley2009comparing}, \textit{Hurley} and \textit{Rickard} have comparatively evaluated sixteen different common metrics of sparsity, and their result strongly recommends to deploy the Gini Index, which is normalized to $[0,1]$, sensitive to both the density and the relative level of outliers, and invariant to scaling or cloning. Given an sequence $\mathbf{c}=\{c_0,c_1,\dots,c_{N-1}\}$, we can define its \textit{sorting in the rising order} as $\overrightarrow{\mathbf{c}}=[c_{(0)},c_{(1)},\dots,c_{(N-1)}]$ that
\begin{equation}
	c_{(m)}\le c_{(n)}\quad\forall 0\le m<n\le N-1,
\end{equation}
and the Gini index of $\mathbf{c}$ can be calculated then as
\begin{equation}
	S(\mathbf{c})=1-2\sum\limits_{k=0}^N-1\frac{c_{(k)}}{||\mathbf{c}||_1}\left(\frac{N-k+\frac{1}{2}}{N}\right),
\end{equation}
where $||\bullet||_1$ is the Manhattan  norm. 

As we are interested in the magnitude of noise samples $\mathbf{x}$ rather than the raw amplitude, here we invoke the Gini Index of $|\mathbf{x}|$ instead of $\mathbf{x}$ to realize a \textit{sparsity-sensitive initialization} (SSI):
\begin{equation}
	T_0=10{S(|\mathbf{x}|)}\tilde{\sigma}_X=14.826{S\left(|\mathbf{x}|\right)}\times\textrm{MAD}(X).
\end{equation}

\section{Numerical Simulations}\label{sec:simulation}
To verify and evaluate our proposed approach, we carried out MATLAB simulations. As test input, BG noise sequences with length of $N=1\times10^5$ samples were generated with different parameters: $\sigma_1^2=1$, $\sigma_2^2\in\{10^2,10^3,10^4,10^5,10^6\}$, $\rho\in\{1\times10^{-4},3\times10^{-4},1\times10^{-3},3\times10^{-3},1\times10^{-2}\}$. For each unique specification, 100 times of Monte-Carlo test were executed to obtain the average impulse detection error rates of Type I and Type II. To evaluate the converging performance, we also recorded the average number of loops to converge. For reference, we also tested the performances of ITS when using mean absolute deviation for the Gaussian estimation instead of the RGE, and when using simple "three-sigma rule" for the initialization instead of our SSI. The results are listed in Tabs.~\ref{tab:type1err}-\ref{tab:loops}.

\begin{table}[!h]
	\centering
	\caption{Type I Error Rate of the ITS Algorithm}
	\label{tab:type1err}
	\begin{tabular}{c*{5}{|c}}
		\toprule[2px]
		\diagbox{$\sigma_2^2$}{$\rho$} & $1\times10^{-4}$ & $3\times10^{-4}$ & $1\times10^{-3}$ & $3\times10^{-3}$ & $1\times10^{-2}$\\\hline
		
						& \textbf{2.79E-5} & \textbf{2.38E-5} & \textbf{2.98E-5} & \textbf{2.33E-5}& \textbf{1.25E-5}\\
		$10^{2}$ & \textit{9.5E-6} & \textit{1.43E-5} & \textit{3.56E-5} & \textit{1.001E-4} & \textit{3.219E-4}\\
						& 0.2135 & 0.2124 & 0.2086 & 0.1982 & 0.1670\\\hline
						
						& \textbf{9.0E-6} & \textbf{1.20E-5} & \textbf{1.17E-5} & \textbf{8.8E-6}& \textbf{4.2E-6}\\
		$10^{3}$ & \textit{3.7E-6} & \textit{5.2E-6} & \textit{1.18E-5} & \textit{3.22E-5} & \textit{7.62E-5}\\
						& 0.2121 & 0.2083 & 0.1956 & 0.1652 & 0.0984\\\hline
						
						& \textbf{8.0E-6} & \textbf{3.0E-6} & \textbf{3.2E-6} & \textbf{4.0E-6}& \textbf{3.2E-6}\\
		$10^{4}$ & \textit{1.0E-6} & \textit{0} & \textit{4.4E-6} & \textit{6.8E-6} & \textit{1.48E-5}\\
						& 0.2076 & 0.1963 & 0.1613& 0.0999& 0.0265\\\hline
						
						& \textbf{3.2E-6} & \textbf{1.0E-6} & \textbf{1.7E-6} & \textbf{1.7E-6}& \textbf{2.0E-6}\\
		$10^{5}$ & \textit{1.0E-6} & \textit{0} & \textit{1.0E-6} & \textit{2.5E-6} & \textit{4.4E-6}\\
						& 0.1955 & 0.1647 & 0.0976 & 0.0299 & 0.0036\\\hline
						
						& \textbf{1.0E-6} & \textbf{0} & \textbf{0} & \textbf{0}& \textbf{2.3E-6}\\
		$10^{6}$ & \textit{0} & \textit{0} & \textit{0} & \textit{0} & \textit{1.7E-6}\\
						& 0.1646 & 0.1005 & 0.0354 & 0.0098 & 0.0025\\\hline
						
		\multicolumn{6}{l}{Legend: \textbf{RGE and SSI}; \textit{only SSI}; only RGE}\\
		\bottomrule[2px]
	\end{tabular}
\end{table}

\begin{table}[!h]
	\centering
	\caption{Type II Error Rate of the ITS Algorithm}
	\label{tab:type2err}
	\begin{tabular}{c*{5}{|c}}
		\toprule[2px]
		\diagbox{$\sigma_2^2$}{$\rho$} & $1\times10^{-4}$ & $3\times10^{-4}$ & $1\times10^{-3}$ & $3\times10^{-3}$ & $1\times10^{-2}$\\\hline
		
						& \textbf{0.3863} & \textbf{0.3680} & \textbf{0.3292} & \textbf{0.3367}& \textbf{0.3453}\\
		$10^{2}$ & \textit{0.3720} & \textit{0.3462} & \textit{0.3227} & \textit{0.3078} & \textit{0.2803}\\
						& 0.1214 & 0.1205 & 0.1046 & 0.1021 & 0.1089\\\hline
						
						& \textbf{0.1830} & \textbf{0.1292} & \textbf{0.1207} & \textbf{0.1154}& \textbf{0.1264}\\
		$10^{3}$ & \textit{0.1678} & \textit{0.1366} & \textit{0.1121} & \textit{0.1060} & \textit{0.0995}\\
						& 0.0556 & 0.0496 & 0.0347 & 0.0351 & 0.0417\\\hline
						
						& \textbf{0.0773} & \textbf{0.0498} & \textbf{0.0426} & \textbf{0.0425}& \textbf{0.0438}\\
		$10^{4}$ & \textit{0.0987} & \textit{0.0586} & \textit{0.0393} & \textit{0.0383} & \textit{0.0354}\\
						& 0.0303 & 0.0204 & 0.0158 & 0.0162 & 0.0192\\\hline
						
						& \textbf{0.0413} & \textbf{0.0228} & \textbf{0.0197} & \textbf{0.0164}& \textbf{0.0146}\\
		$10^{5}$ & \textit{0.0529} & \textit{0.0257} & \textit{0.0172} & \textit{0.0148} & \textit{0.0125}\\
						& 0.0182 & 0.0140 & 0.0075 & 0.0082 & 0.0096\\\hline
						
						& \textbf{0.0143} & \textbf{0.0095} & \textbf{0.0091} & \textbf{0.0066}& \textbf{0.0053}\\
		$10^{6}$ & \textit{0.0236} & \textit{0.0150} & \textit{0.0072} & \textit{0.0055} & \textit{0.0046}\\
						& 0.0178 & 0.0033 & 0.0055 & 0.0049 & 0.0043\\\hline
						
		\multicolumn{6}{l}{Legend: \textbf{RGE and SSI}; \textit{only SSI}; only RGE}\\
		\bottomrule[2px]
	\end{tabular}
\end{table}

\begin{table}[!h]
	\centering
	\caption{Average Converging Loops of the ITS Algorithm}
	\label{tab:loops}
	\begin{tabular}{c*{5}{|c}}
		\toprule[2px]
		\diagbox{$\sigma_2^2$}{$\rho$} & $1\times10^{-4}$ & $3\times10^{-4}$ & $1\times10^{-3}$ & $3\times10^{-3}$ & $1\times10^{-2}$\\\hline
		
						& \textbf{3.41} & \textbf{3.90} & \textbf{3.13} & \textbf{2.67}& \textbf{8.66}\\
		$10^{2}$ & \textit{1.59} & \textit{2.94} & \textit{10.07} & \textit{37.72} & \textit{193.47}\\
						& 2.7 & 3 & 3.4 & 4.4 & 5.9\\\hline
		
						& \textbf{4.12} & \textbf{3.60} & \textbf{3.02} & \textbf{2.01}& \textbf{15.28}\\
		$10^{3}$ & \textit{1.09} & \textit{1.41} & \textit{4.03} & \textit{16.64} & \textit{118.71}\\
						& 2.7 & 2.8 & 3.6 & 7.2 & 20.5\\\hline
		
						& \textbf{3.69} & \textbf{2.85} & \textbf{3.02} & \textbf{1.85}& \textbf{15.13}\\
		$10^{4}$ & \textit{1.06} & \textit{1.03} & \textit{2.15} & \textit{11.06} & \textit{98.02}\\
						& 2.6 & 3.2 & 6& 37.5& 978.5\\\hline
		
						& \textbf{2.58} & \textbf{1.52} & \textbf{1.12} & \textbf{2.20}& \textbf{8.59}\\
		$10^{5}$ & \textit{1.02} & \textit{1.05} & \textit{1.86} & \textit{9.45} & \textit{95.35}\\
						& 2.8 & 3.4 & 33& 1133.3 & 772.9\\\hline
		
						& \textbf{1.44} & \textbf{1.04} & \textbf{1.10} & \textbf{1.79}& \textbf{4.18}\\
		$10^{6}$ & \textit{1.00} & \textit{1.07} & \textit{1.92} & \textit{9.46} & \textit{94.52}\\
						& 4.1 & 2.5 & 169.4 & 462.4 & 546.6\\\hline
		
		\multicolumn{6}{l}{Legend: \textbf{RGE and SSI}; \textit{only SSI}; only RGE}\\
		\bottomrule[2px]
	\end{tabular}
\end{table}

 As the numerical results show, our approach of ITS with SSI and RGE appears satisfying in most test cases, providing low error rates and a good converging performance. In the cases with low impulse power $\sigma_2^2$, the Type II error rate is relatively high because some samples are determined as impulses by the Bernoulli process, but assigned with only low amplitude by the Gaussian process, and are hard to detect. An instance of this phenomenon is illustrated in Fig.~\ref{fig:misdetections}. Besides, it also worths to note that both the SSI method and the RGE method are clearly efficient in suppressing the detection error rate and boosting the convergence.
\begin{figure}
	\centering
	\includegraphics[width=.49\textwidth]{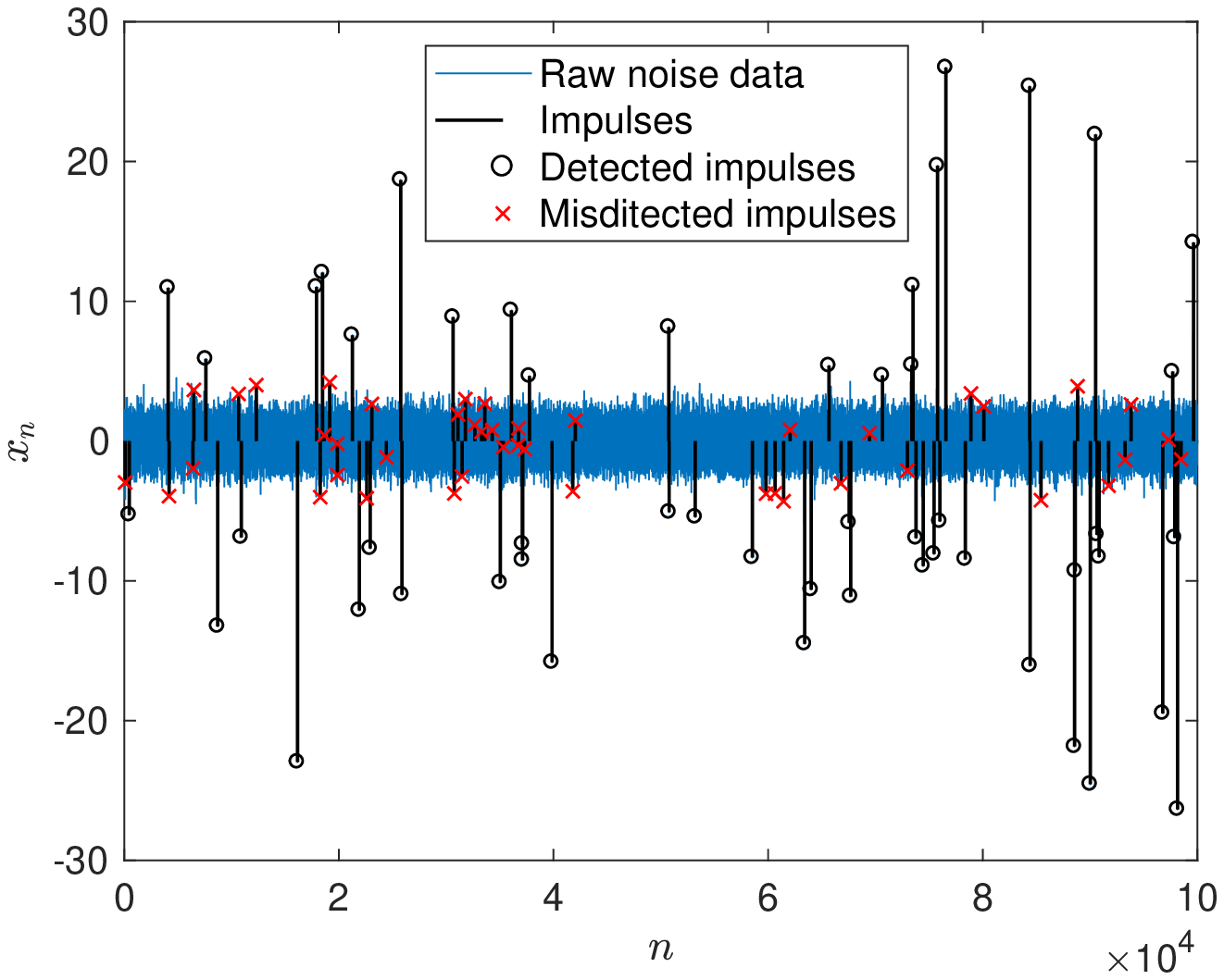}
	\caption{A sample detection result of the ITS algorithm with RGE and SSI. $\sigma_1^2=1$, $\sigma_2^2=10^2$, $\rho=1\times10^{-3}$. All misdetected impulses are low in amplitude, not outlying from the background noise.}
	\label{fig:misdetections}
\end{figure}

\section{Conclusion}\label{sec:conclusion}
So far, in this paper we have presented a novel approach of blind impulse detection for Bernoulli-Gaussian noise in underspread channels, which requires no a priori knowledge about the BG model parameters. The proposed ITS algorithm has been mathematically derived as equivalent to the classical SMLR algorithm under the condition of underspread channel, with a significantly reduced computational complexity. The MAD-based robust Gaussian estimation and a new sparsity-sensitive threshold initialization technique have been applied, in order to enhance the detection accuracy and boost the convergence. The efficiency of our approach has been verified through numerical simulations.

As future work, it is of interest to generalize the proposed method for frequency-selective channels, in order to assist estimation of channel impulse response under presence of impulsive noises, which can benefit various applications such as physical layer security\cite{ambekar2012improving}, adaptive spreading codes\cite{ambekar2010channel} and radio channel integrity monitoring\cite{weinand2017application}.

\section*{Acknowledgment}
This work was partly supported  by the German Ministry of Education and Research (BMBF) under grant 16KIS0263K (HiFlecs). The authors would like to acknowledge the contributions of their colleagues. This information reflects the consortium's view, but the consortium is not liable for any use that may be made of any of the information contained therein.

%\bibliographystyle{IEEEtran}
%\bibliography{references}

\end{document}